\documentclass[a4paper,10pt]{article}

\usepackage{amssymb}
\usepackage{amsthm}
\usepackage{subfigure}
\usepackage{inputenc}
\usepackage{graphicx}
\usepackage{float}
\usepackage{natbib}
\usepackage{amsfonts}
\usepackage[tbtags]{amsmath}
\usepackage{authblk}
\begin{document}

\title{Multi-agent learning using Fictitious Play and Extended Kalman Filter}
\author[1]{Michalis Smyrnakis\thanks{michalis.smyrnakis@manchester.ac.uk}}
\affil[1]{Complex Systems and Statistical Physics Group, School of Physics and Astronomy, University of Manchester, UK}
\maketitle
\begin{abstract}
Decentralised optimisation  tasks are important components of multi-agent systems. These tasks can be interpreted as n-player potential games: therefore game-theoretic learning algorithms can be used to solve decentralised optimisation tasks. Fictitious play is the canonical example of  these algorithms. Nevertheless fictitious play implicitly assumes that players have stationary strategies. We present a novel variant of fictitious play where players predict their opponents' strategies using Extended Kalman filters and use their predictions to update their strategies.

We show that in 2 by 2 games with at least one pure Nash equilibrium and in potential games where players have two available actions, the proposed algorithm converges to the pure Nash equilibrium. The performance of the proposed algorithm was empirically tested,  in two strategic form games and an ad-hoc sensor network surveillance problem. The proposed algorithm performs better than the classic fictitious play algorithm in these games and therefore improves the performance of game-theoretical learning in decentralised optimisation. 
\end{abstract}

\newtheorem{theor}{Theorem}   
\newtheorem{prop}{Proposition}
\section{Introduction}
Recent advance in technology render decentralised optimisation a crucial component of many applications of multi agent systems and decentralised control. Sensor networks \citep{sn}, traffic control \citep{tf} and  scheduling problems \citep{sc} are some of the tasks where decentralised optimisation can be used. These tasks share common characteristics such as large scale, high computational complexity and communication constraints that make a centralised solution intractable. It is well known that many decentralised optimisation tasks can be cast as potential games \citep{wlu,autonomous}, and the search of an optimal solution can be seen as the task of finding Nash equilibria in a game. Thus it is feasible to use iterative learning algorithms from game-theoretic literature to solve decentralised optimisation problems.

A game theoretic learning algorithm with proof of convergence in certain kinds of games is fictitious play \citep{learning_in_games,fp4}. It is a learning process where players choose an action that maximises their expected rewards according to the beliefs they maintain about their opponents' strategies.The players update their beliefs about their opponents' strategies after
observing their actions. Even though fictitious play converges to Nash equilibrium, this convergence can be very slow. This is because it implicitly assumes that other players use a fixed strategy in the whole game. \citet{cj} addressed this problem by representing the fictitious play process as a state space model and by using particle filters to predict opponents' strategies. The drawback of this approach is the computational cost of the particle filters that render difficult the application of this method in real time applications.

The alternative that we propose in this article is to use instead of particle filters, extended Kalman filters (EKF) to predict opponents' strategies. Therefore the proposed algorithm has smaller computational cost than the particle filter variant of fictitious play algorithm that proposed by \citet{cj}. We show that the EKF fictitious play algorithm converges to a pure Nash equilibrium, in 2 by 2  games with at least one pure Nash equilibrium and in potential games where players have two available actions. We also empirically observe, in a range of games, that the proposed algorithm needs less iterations than the classic fictitious play to converge to a solution. Moreover in our simulations, the  proposed algorithm converged to a solution with higher reward than the classic fictitious play algorithm. 

The remainder of this paper is organised as follows. We start with a brief description of game theory, fictitious play and extended Kalman filters. Section \ref{new_algorithm} introduces the proposed algorithm that combines fictitious play and extended Kalman filters. The convergence results we obtained are presented in Section \ref{theory}. In Section \ref{parameters} we propose some indicative values for the EKF algorithm parameters. Section \ref{simulation} presents the simulation results of EKF fictitious play in  a 2$\times$2 coordination game, a three player climbing hill game and an ad-hoc sensor network surveillance problem. In the final section we present our conclusions. 

\section{Background}
\label{background}
In this section we introduce some definition from game theory that we will use in the rest of this article and the relation between potential games and decentralised optimisation.  We also briefly present the classic fictitious play algorithm and the extended Kalman filter algorithm.
\subsection{Game theory definitions}
We consider a game $\Gamma$ with $\mathbb{I}$ players, where each player $i$,$i=1,2,\ldots,\mathbb{I}$, choose his action, $s^{i}$, from a finite discrete set $S^{i}$. We then can define the joint action that is played in a game as the set product $S=\times_{i=1}^{i=\mathbb{I}}S^{i}$. Each Player $i$ receive a reward, $u^{i}$, after choosing an action . The reward is a map from the joint action space to the real numbers,  $u^{i}:S \rightarrow R$. We will often write
$s=(s^{i},s^{-i})$, where $s^{i}$ is the action of Player $i$ and $s^{-i}$ is the joint action of Player $i$'s opponents. When players select their actions using a probability distribution they use mixed strategies. The mixed strategy of a player $i$, $\sigma^{i}$, is an element of the set $\Delta^{i}$, where $\Delta^{i}$ is the set of all the probability distributions over the action space $S^{i}$. The joint mixed strategy, $\sigma$, is then an element of $\Delta=\times_{i=1}^{i=\mathbb{I}}\Delta^{i}$. Analogously to the joint actions we will write  $\sigma=(\sigma^{i},\sigma^{-i})$. In the special case where the players choose an action with probabiity one we will say that players choose their actions using pure strategies. The expected utility a player $i$ will gain if he chooses a strategy $\sigma^{i}$ (resp.\ $s^{i}$), when his opponents choose the joint strategy $\sigma^{-i}$ is $u^{i}(\sigma^{i},\sigma^{-i})$ (resp.\ $u^{i}(s^{i},\sigma^{-i})$). 

A common decision rule in game theory is best response (BR). The best response is defined as the action that maximizes players' expected utility given their opponents' strategies. Thus for a specific opponents' strategy  $\sigma^{-i}$ we evaluate the best response as:
\begin{equation}
BR^{i}(\sigma^{-i})= \mathop{\rm argmax}_{s^{i} \in S} \quad
u^{i}(s^{i},\sigma^{-i})
\label{eq:br}
\end{equation} 

\citet{nash} showed that every game has at least one equilibrium, which is a fixed point of
the best response correspondence, $\sigma^{i} \in BR(\sigma^{-i})$. Thus when a joint mixed strategy $\hat{\sigma}$ is a Nash equilibrium then:
\begin{equation}
u^{i}(\hat{\sigma}^{i},\hat{\sigma}^{-i})\geq u^{i}(s^{i},\hat{\sigma}^{-i}) \qquad \textrm{for all }  s^{i} \in S^{i}
\label{eq:nashutil}
\end{equation}

Equation \ref{eq:nashutil} implies that if a strategy
$\hat{\sigma}$ is a Nash equilibrium then it is not possible for a player to
increase his utility by unilaterally changing his strategy. When all the players in a game select their actions using pure strategies then the
equilibrium actions are referred as pure strategy Nash equilibria. A pure equilibrium is strict if
each player has a unique best response to his opponents actions.
\subsection{Decentralised optimisation tasks as potential games}
A class of games that are of particular interest in multi agent systems and decentralised optimisation tasks are potential games, because of their utility structure. In particular in order to be able to solve an optimisation task decentrally the local functions should have similar characteristics with the global function that we want to optimise. This suggests that an action which improves or reduces the utility of an individual should respectively increase or reduce the global utility. Potential games have this property, since the potential function (global function) depict the changes in the players' payoffs (local functions) when they unilaterally change their
actions. More formally we can write 
\begin{equation*}
u^{i}(s^{i},s^{-i})-u^{i}(\tilde{s^{i}},s^{-i})=
\phi(s^{i},s^{-i})-\phi(\tilde{s^{i}},s^{-i})
\end{equation*}
where $\phi$ is a potential function and the above equality stands
for every player $i$, for every action $s^{-i}\in S^{-i}$, and for
every pair of actions $s^{i}$, $\tilde{s^{i}} \in S^{i}$, where
$S^{i}$ and $S^{-i}$ represent the set of all available actions
for Player $i$ and his opponents respectively.  Moreover potential games has at least one pure Nash equilibrium, hence there is at least one joint action $s$ where no player can increase their reward, therefore the potential function, through a unilateral deviation.

It is feasible to choose an appropriate form of the agents' utility function in order for the global utility to act as a potential of the system. Wonderful life utility is a utility function that introduced by \citet{wlu} and applied by \citet{autonomous} to formulate distributed optimisation tasks  as potential games. Player $i$'s utility, when wonderful life utility is used, can be defined as the difference between the global utility $u_{g}$ and the utility of the system when a reference action is used as player's $i$ action. More formally when player $i$ chooses an action $s^{i}$ we write
\begin{equation*}
u^{i}(s^{i})=u_{g}(s^{i},s^{-i})-u_{g}(s^{i}_{0},s^{-i})
\label{eq:wlu}
\end{equation*}
where $s^{i}_{0}$ denotes the reference action of player $i$. Hence the decentralised optimisation
problem can be cast as a potential game and any algorithm that is proved to
converge to a Nash equilibrium of a potential game, which is a local or the global optimum of the optimisation problem, will converge to a joint action from which no player
can increase the global reward through unilateral deviation.

\subsection{Fictitious play}
Fictitious play \citep{brown_fict}, is a widely used learning technique in game theory. In fictitious play each player chooses
his action according to the best response to his beliefs about his opponents' joint mixed strategy $\sigma^{-i}$.

Initially each player has some prior beliefs about the strategy
that each of his opponents uses to choose an action based on a weight function $\kappa_{t}$. The
players, after each iteration, update the weight function and therefore their beliefs about their opponents' strategies and
play again the best response according to their beliefs. More
formally in the beginning of a game Player $i$ maintains some arbitrary non-negative initial weight functions
$\kappa_{0}^{j}$, $\forall j\in[1, \mathbb{I}] \backslash \{i\}$, that are updated using the formula:
\begin{equation*}
\kappa_{t}^{j}(s^{j}) = \kappa_{t-1}^{j}(s^{j})+\mathfrak{I}_{s^j_t=s^j}
\label{eq:kappa}
\end{equation*}
for each $j$, where
$\mathfrak{I}_{s^j_t=s^j}=\left\{\begin{array}{cl}1&\mbox{if
$s^j_t=s^j$}\\0&\mbox{otherwise.}\end{array}\right.$. \\The mixed
strategy of opponent $j$ is estimated from the following formula:

\begin{equation}
\sigma_{t}^{j}(s^{j})=\frac{\kappa^j_{t}(s^j)}{\sum_{s' \in S^{j}}\kappa^j_{t}(s')}.
\label{eq:fp1p}
\end{equation}

%\begin{equation}
%\sigma_{t}^{j}(s^{j})=\left(1-\frac{1}{t^j}\right)\sigma^j_{t-1}(s^j) + \frac{1}{t^j}\mathfrak{I}_{s^j_t=s^j}
%\label{eq:fp}
%\end{equation}
%where $t^j=t+\sum_{s^j\in S^j}\kappa^j_0(s^j)$.

Player $i$ based on his beliefs about his opponents' strategies,  chooses the action which maximises his expected payoffs. When player $i$ uses equation (\ref{eq:fp1p}) to update the beliefs about his opponents' strategies he treats the environment of the game as stationary and implicitly assumes that the actions of the
players are sampled from a fixed probability distribution. Therefore
the recent observations have the same weight as the initial ones.
This approach leads to poor adaptation when the other players
choose to change their strategies.

\subsection{Fictitious play as a state space model}
We follow \citet{cj} and we will represent fictitious play process as a state-space model. 
%\begin{figure}
%\centering
% \includegraphics[scale=0.5]{actions3.jpg}
%\caption{State space representation of fictitious play process.}
%\label{fig:propensities}
%\end{figure}
%Figure \ref{fig:propensities} is a graphical representation of the process that players are using to choose their actions as in \citet{cj}, where $Q_{t}$, $\sigma_{t}$ and $s_{t}$
%represent the propensity, the strategy and the action of a Player $i$ at time $t$ respectively. 
According to this state space model each player has a propensity $Q_{t}^{i}(s^{i})$ to play each of his available actions $s^{i} \in S^{i}$, and then he  forms his strategy based on these propensities. Finally he chooses his actions based on his strategy and the best response decision rule. Because players have no information about the evolution of their opponents' propensities, and under the assumption that the changes in propensities are small from one iteration of the game to another, we model propensities using a Gaussian autoregressive prior on all propensities. 
We set $Q_{0}\sim N(0,I)$ and recursively update the value of  $Q_{t}$ according to the value of $Q_{t-1}$ as follows:
\begin{equation*}
Q(s_{t})=Q(s_{t-1})+\eta_{t}
\label{eq:propen}
\end{equation*}
where $\eta_{t}\sim N(0,\chi^{2}I)$. 
The action of a player then is related to his propensity by the following sigmoid equation for every $s^{i} \in S^{i}$
\begin{equation*}
 s^{i}=\frac{e^{(Q^{i}(s^{i})/\tau)}}{\sum_{\tilde{s} \in S^{i}}e^{(Q_{t}(\tilde{s})/\tau)}}.
\end{equation*}
Therefore players will assume that  at every iteration $t$ their opponents have a different strategy $\sigma_{t}$. %Hence they should estimate their opponents' strategies based on this assumption instead of using equation (\ref{eq:fp1p}). 

\subsection{Kalman filters and Extended Kalman filters}
Our objective is to estimate player $i$'s opponent propensity and thus to estimate the marginal probability $p(Q_{t},s_{1:t})$. This objective can be represented as a Hidden Markov Model (HMM). HMMs are used to predict the value of an unobserved variable $x_{t}$, the hidden state,  using the observations of another variable $z_{1:t}$.  There are two main assumptions in the HMM representation. The former one is that the probability of being at any state $x_{t}$ at time $t$ depends only at the state of time $t-1$, $x_{t-1}$. The latter one is that an observation at time $t$ depends only on the current state $x_{t}$. One of the most common methods to estimate $p(x_{1:t},z_{1:t})$ is Kalman filters and its variations. Kalman filter \citep{kalman} is based on two assumptions, the first is that the state variable is Gaussian. The second is that the observations are the result of a linear combination of the state variable. Hence Kalman filters can be used in cases which are represented as the following state space 
model:
\begin{equation*}
\label{eq:kf_state}
\begin{split}
x_t= &  Ax_{t-1}+\xi_{t-1}  \textrm{   hidden layer} \\
y_{t}= & Bx_{t}+\zeta_{t} \textrm{   	observations} 
\end{split}
\end{equation*}
where $\xi_t$ and $\zeta_t$ follow a zero mean normal distribution with covariance matrices $\Xi=q_{t}I$ and $Z=r_{t}I$ respectively, and $A$, $B$ are linear transformation matrices. When the distribution of the state variable $x_t$ is Gaussian then $p(x_{t}|y_{1:t})$ is also a Gaussian distribution, since $y_{t}$ is a linear combination of $x_{t}$. Therefore it is enough to estimate its mean and variance to fully characterise $p(x_{t}|y_{1:t})$. 

Nevertheless in the state space model we want to implement, the relation between Player $i$'s opponent propensity and his actions is not linear. Thus we should use a more general form of state space model such as:  
\begin{equation}
\label{eq:ekf_state}
\begin{split}
x_t&=f(x_{t-1})+\xi_{t} \\
y_{t}&=h(x_{t})+ \zeta_{t} 
\end{split}
\end{equation}
where $\xi_{t}$ and $\zeta_{t}$ are the hidden and observation state noise respectively, with zero mean and covariance matrices $\Xi=q_{t}I$ and $Z=r_{t}I$ respectively. The distribution of $p(x_{t}|y_{1:t})$ is not a Gaussian distribution because $f(\cdot)$ and  $h(\cdot)$ are non-linear functions. A simple method to overcome this shortcoming is to use a first order Taylor expansion to approximate the distributions of the sate space model in (\ref{eq:ekf_state}). In particular we let $x_{t}=m_{t-1}+\epsilon$, where $m_{t}$ denotes the mean of $x_{t}$ and $\epsilon \sim N(0,P)$. We can rewrite (\ref{eq:ekf_state}) as: 
\begin{equation}
\label{eq:ekf_taylor}
\begin{split}
x_t&=f(m_{t-1}+\epsilon)+w_{t-1}=f(m_{t-1})+F_{x}(m_{t-1})\epsilon +\xi_{t-1}\\
y_{t}&=h(m_{t}+\epsilon)+\zeta_{t}=h(m_{t})+H_{x}(m_{t})\epsilon+\zeta_{t}
\end{split}
\end{equation}
where $F_{x}(m_{t-1})$ and $H_{x}(m_{t})$ is the Jacobian matrix of $f$ and $h$ evaluated at $m_{t-1}$ and $m_{t}$, respectively. If we use the transformations in (\ref{eq:ekf_taylor}) then $p(x_{t}|y_{1:t})$ is a Gaussian distribution.

Since $p(x_{t}|y_{1:t})$ is a Gaussian distribution to fully characterise it we need to evaluate its mean and its variance. The EKF process \citep{ekf1,ekf2} estimates this mean and variance in two steps the prediction and the update step. In the prediction step at any iteration $t$ the distribution of the state variable is estimated based on all the observations until time $t-1$, $p(x_{t}|y_{1:t-1})$. The distribution of $p(x_{t}|y_{1:t-1})$ is Gaussian and we will denote its mean and variance  as  $m_{t}^{-}$ and $P_{t}^{-}$ respectively. During the update step the estimation of the prediction step is corrected in the light of the new observation at time $t$, so we estimate $p(x_{t}|y_{1:t})$. This is also a Gaussian distribution and we will denote its mean and variance  as  $m_{t}$ and $P_{t}$ respectively.

The prediction and the update steps of the EKF process \citep{ekf1,ekf2} to estimate the mean and the variance of $p(x_{t}|y_{1:t-1})$ and $p(x_{t}|y_{1:t})$ respectively are the following: \\
\textbf{Prediction Step}
\begin{align}
\label{eq:pdstep}
m_{t}^{-}= &f(m_{t-1}) \nonumber \\
P_{t}^{-}=&F(m_{t-1})P_{t-1}F(m_{t-1})+\Xi_{t-1} \nonumber
\end{align}
where the $j,j'$ element of $F(m_{t})$ is defined as
\begin{equation*}
[F(m_{t}^{-})]_{j,j'}=\frac{\partial f(x_{j},r)}{\partial x_{j'}}\arrowvert_{x=m_{t}^{-}, q=0}	
\end{equation*}
\textbf{Update Step}
\begin{eqnarray}
\label{eq:updatestep}
v_t&=&z_{t}-h(m_{t}^{-}) \nonumber \\
S_{t}&=&H(m_{t}^{-})P_{t}^{-}H^{T}(m_{t}^{-})+Z \nonumber \\
K_{t}&=&P_{t}^{-}H^{T}(m_{t}^{-})S_{t}^{-1} \nonumber \\
m_{t}&=&m_{t}^{-}+K_{t}v_{t}\nonumber \\
P_{t}&=&P_{t}^{-}-K_{t}S_{t}K_{t}^{T} \nonumber
\end{eqnarray}
where $z_{t}$ is the observation vector (with 1 in the entry of the observed action and 0 everywhere else) and the $j,j'$ element of $H(m_{t})$ is defined as:  
\begin{equation*}
[H(m_{t}^{-})]_{j,j'}=\frac{\partial h(x_{j},r)}{\partial x_{j'}}\arrowvert_{x=m_{t}^{-}, r=0}	
\end{equation*}

\section{Fictitious play and EKF}
\label{new_algorithm}
For the rest of this paper we will only consider inference over a single opponent mixed
strategy in fictitious play. Separate estimates will be formed
identically and independently for each opponent. We therefore
consider only one opponent, and we drop all dependence on player $i$, and
write $s_{t}$, $\sigma_{t}$ and $Q_{t}$ for Player $i$'s opponent's action, strategy and propensity respectively. Moreover for any vector $x$, $x[j]$ will denote the $j_{th}$ element of the vector and for any matrix $y$, $y[i,j]$ will denote the $(i,j)_{th}$ element of the matrix.

We can use the following state space model to describe the fictitious play process:
\begin{align}
\label{eq:fpekf}
Q_{t}&=Q_{t-1}+\xi_{t-1} \nonumber \\ 
s_{t}&= h(Q_{t})+\zeta_{t} \nonumber
\end{align}
where $\xi_{t-1} \sim N(0,\Xi)$, is the noise of the state process and $\zeta_{t}$ is is the error of the observation state with zero mean and covariance matrix $Z$, which occurs because we approximate a discrete process like best responses, equation (\ref{eq:br}), using a continuous function $h(\cdot)$. Hence we can combine the EKF with fictitious play as follows.
At time $t-1$ Player $i$ has an estimation of his opponent's propensity using a Gaussian distribution with mean $m_{t-1}$ and variance $P_{t-1}$, and has observed an action $s_{t-1}$. Then at time $t$ he uses EKF prediction step to estimate his opponent's propensity. The mean and variance of $p(Q_{t}|s_{1:t-1})$ of the opponent's propensity approximation are:
\begin{align}
%\label{eq:fppred}
m^{-}_{t}=m_{t-1} \nonumber \\
P_{t}^{-}=P_{t-1}+\Xi \nonumber
\end{align}
Player $i$ then evaluates his opponents strategies using his estimations as:
\begin{equation}
 \sigma_{t}(s_{t})=\frac{exp(m_{t}^{-}[s_{t}] / \tau)}{\sum_{\tilde{s} \in S}exp(m_{t}^{-}[\tilde{s}]/\tau)}.
\label{eq:strategies}
\end{equation}
where $m_{t}^{-}[s_{t}]$ is the mean of Player $i$'s estimation about the propensity of his opponent to play action $s_{t}$. Player $i$ then uses the estimation of his opponent strategy , equation  (\ref{eq:strategies}), and best responses, equation (\ref{eq:br}), to choose an action. After observing the opponent's action $s_{t}$, Player $i$ correct his estimations about his opponent's propensity using the update equations of EKF process. The update equations are:

\begin{eqnarray}
v_t&=&z_{t}-h(m_{t}^{-}) \nonumber \\
S_{t}&=&H(m_{t}^{-})P_{t}^{-}H^{T}(m_{t}^{-})+Z \nonumber \\
K_{t}&=&P_{t}^{-}H^{T}(m_{t}^{-})S_{t}^{-1} \nonumber \\
m_{t}&=&m_{t}^{-}+K_{t}v_{t} \nonumber \\
P_{t}&=&P_{t}^{-}-K_{t}S_{t}K_{t}^{T} \nonumber
\label{eq:fpupdate}
\end{eqnarray}
where $h=\frac{exp(Q_{t}[s^{'}]/\tau)}{\sum_{\tilde{s} \in S} exp(Q_{t}[\tilde{s}]/\tau)}$, and $\tau$ is a temperature parameter. The Jacobian matrix $H(m_{t}^{-})$ is defined as \\ $[H(m_{t}^{-})]_{j,j'}=\left\{\begin{array}{cl} \frac{\sum_{j \neq j'}\exp(m_{t}^{-}[j])\exp(m_{t}^{-}[j'])}{(\sum_{j}\exp(m_{t}^{-}[j]))^2}&\mbox{if
$j=j'$}\\- \frac{\exp(m_{t}^{-}[j])\exp(m_{t}^{-}[j'])}{(\sum_{j}\exp(m_{t}^{-}[j]))^2} &\mbox{if $j \neq$ j'}\end{array}\right.$. 

Table \ref{skata} summarises the fictitious play algorithm when EKF is used to predict opponents strategies. 
\begin{table}[ht]
\begin{center}
\begin{tabular}{p{12cm}}
\hline
\hline
At time $t$
\begin{enumerate}
\item Player $i$ maintains some estimations about his opponents propensity up to time $t-1$, $p(Q_{t-1}|s{1:t-1})$. Thus he has an estimation of the mean $m_{t-1}$ and the covariance $P_{t-1}$ of this distribution.

\item Then Player $i$ is updating his estimations about his opponents propensities $p(Q_{t}|s{1:t-1})$ using equations, $m_{t}^{-}=m_{t-1}$,  $P_{t}^{-}=P_{t-1}+W_{t-1}$.

\item Based on the weights of step 1 each player updates his beliefs about his
opponents strategies using $\sigma_{t}^{j}(s^{j})=\frac{exp(m_{t}^{-}(j)/\tau)}{\sum_{j'}exp(m_{t}^{-}(j)/\tau)}$.

\item Choose an action based on the beliefs of step 3 according to best
response.

\item Observe opponent's action $s_{t}$.

\item Update the propensities estimates using $	m_{t}=m_{t}^{-}+K_{t}v_{t}$ and \mbox{$P_{t}=P_{t}^{-}-K_{t}S_{t}K_{t}^{T}$}.

\item set t=t+1
\end{enumerate}
\\
\hline
\hline
\end{tabular}
\caption{EKF Fictitious Play algorithm}
\label{skata}
\end{center}
\end{table}

\section{Theoretical Results}
\label{theory}
In this section we present the convergence results we obtained for games with at least one pure Nash equilibrium and players who have 2 available actions, $s=(1,2)$. We will denote as $-s$ the action that a player does not choose, for example if Player $i$'s opponent chooses action 1, $s=1$ and hence $-s=2$. Also we will denote as $m[1]$ and $m[2]$ the estimated means of opponent's propensity  of action 1 and 2 respectively. Similarly $P[1,1]$ and $P[2,2]$ will represent the variance of the propensity's estimation of action 1 and 2 respectively, and $P[1,2],P[2,1]$ their covariance. 

The proposed algorithm has the following two properties: 
\begin{prop}
\label{prop1}
If at iteration $t$ of the EKF fictitious play algorithm, action $s$ is played from Player $i$'s opponent, then the estimation of his opponent propensity to play action $s$ increases, $m_{t-1}[s]<m_{t}[s]$. Also the estimation of his opponent propensity to play action $-s$ decreases, $m_{t-1}[-s]>m_{t}[-s]$ 
\end{prop}

\begin{proof}
The proof of Proposition \ref{prop1} is on Appendix \ref{append1}.
\end{proof}

Proposition \ref{prop1} implies that players, when they use EKF fictitious play, learn their opponent's strategy and eventually they will choose the action that will maximise their reward base on their estimation. Nevertheless there are cases where players may change their action simultaneously and trapped in a cycle instead of converging in a pure Nash equilibrium. As an example we consider the game that is depicted in Table \ref{tab:simcoord}.      

\begin{table}
\centering
\begin{tabular}{|c| c| c|}
\hline
 &L&R\\ \hline
 U& 1,1 & 0,0 \\ \hline
 D& 0,0 & 1,1 \\ \hline
 \end{tabular}
\caption{Simple coordination game}
 \label{tab:simcoord}
\end{table}
This is a simple coordination game with two pure Nash equilibria the joint actions $(U,L)$ and $(D,R)$. In the case were the two players start from joint action $(U,R)$ or $(D,L)$ and they always change their action simultaneously then they will never reach one of the two pure Nash equilibria of the game.

\begin{prop}
\label{prop2}
In a $2 \times 2$ game where the players use EKF fictitious play process to choose their actions, and the variance of the observation state is set to $Z=rI+\epsilon I$, with high probability the two players will not change their action simultaneously infinitely often. We define $\epsilon$ as a random number from normal distribution with zero mean and arbitrarily small covariance matrix, $I$ is the identity matrix. 
\end{prop}

\begin{proof}
The proof of Proposition \ref{prop2} is on Appendix \ref{append2}.
\end{proof}

We should mention here that the reason we set $Z=rI+\epsilon I$ is in order to break any symmetries that occurred because the initialisation of the EKF fictitious play algorithm. Based on Proposition \ref{prop1} and \ref{prop2} we can infer the following propositions and theorems. %The proofs are closely parallel to Propositions 2.1 and 2.2 of \citep{learning_in_games} and Theorem 6.2 of \citep{payton_young}.

\begin{prop}
\label{prop3}
(a) In a game where players have two available actions if $s$ is a Nash equilibrium, and $s$
is played at date $t$ in the process of EKF fictitious play, $s$ is
played at all subsequent dates. That is, strict Nash equilibria are absorbing for the
process of EKF fictitious play. (b) Any pure strategy steady
state of EKF fictitious play must be a Nash equilibrium.
\end{prop}

\begin{proof}
Consider the case where players beliefs $\hat{\sigma}_{t}$, are such that their optimal choices correspond
to a strict Nash equilibrium $\hat{s}$. In EKF fictitious play process players' beliefs are formed
identically and independently for each opponent based on equation (\ref{eq:strategies}). By Proposition \ref{prop1} we know that players' estimations about their opponents' propensities and therefore their strategies, that each player maintains for the other players, will increase for the actions that are included in $\hat{s}$ and will be reduced otherwise. Thus the best response to their beliefs $\hat{\sigma}_{t+1}$ will be again $\hat{s}$ and since $\hat{s}$ is a Nash equilibrium they will not deviate from it. Conversely, if a player remains at a pure strategy profile, then eventually the assessments will become concentrated at that profile, because of Proposition \ref{prop1}, hence if the profile is not
a Nash equilibrium, one of the players would eventually want to deviate.
\end{proof}

\begin{prop}
\label{prop4}
Under EKF fictitious play, if the beliefs over each player's choices converge, the strategy profile corresponding
to the product of these distributions is a Nash equilibrium. 
\end{prop}

\begin{proof}
Suppose that the beliefs of the players at time t, $\sigma_{t}$, converges to some
profile $\hat{\sigma}$. If $\hat{\sigma}$ were not a Nash equilibrium, some player would eventually want
to deviate and the beliefs would also deviate since based on Proposition \ref{prop1} players eventually learn their opponents actions. 
\end{proof}

\begin{theor}
\label{theo1}
The EKF fictitious play process converges  to the Nash equilibrium in $2\times2$ games with at least one pure Nash equilibrium, when the covariance matrix of the observation space error, $Z$, is defined as in Proposition \ref{prop2}, $Z=rI+\epsilon I$. 
\end{theor}

\begin{proof}
We can distinct two possible initial states in the game. In the first players' initial beliefs of the players actions are such that their initial joint action $s_{0}$ is a Nash equilibrium. From Proposition \ref{prop3} and equation (\ref{eq:strategies}) we know that they will play the joint action which is a Nash equilibrium for all the iterations of the game. 

The second case where the initial beliefs of the players are such that their initial joint action $s_{0}$ is not a Nash equilibrium is divided in 2 subcategories. The first include $2 \times 2$ games with only one pure Nash equilibrium. In this case, one of the two players has a dominant action, thus for all the iterations of the game he will choose the dominant action. This action maximises his expected payoff regardless the other player's strategy and thus he will select this action in every iteration of the game. Therefore because of Proposition \ref{prop1} the other player will learn his opponent's strategy and players will choose the joint action which is the pure Nash equilibrium. 

The second category includes $2\times2$ games with 2 pure Nash equilibria, like the simple coordination game that is depicted in Table \ref{tab:simcoord}. In this case players initial joint action $s_{0}=(s^{1},s^{2})$ is not a Nash equilibrium. Then the players will learn their opponent's strategy, Proposition \ref{prop1} and Equation (\ref{eq:strategies}), and they will change their action. We know from Proposition \ref{prop2} that in a finite time with high probability the players will not change their actions simultaneously, and hence they will end up in a joint action that will be one of the two pure Nash equilibria of the game. 
\end{proof}

We can extend the results of Theorem \ref{theo1} in  $n \times 2$ games with a better reply path. A game with a better reply path  can be represented
as a graph were its edges are the join actions of the game $s$ and there is a vertex that connects $s$ with $s'$ iff only one player $i$ can increasing his payoff by changing his action \citep{payton_young}. Potential games have a better reply path.  
\begin{theor}
\label{theo2}
The EKF fictitious play process converges to the Nash equilibrium in $n\times 2$ games with a better reply path when the covariance matrix of the observations space error, $Z$, is $Z=r+\epsilon I$. 
\end{theor}

\begin{proof}
Similarly to the $2 \times 2$ games if the initial beliefs of the players are such that their initial joint action $s_{0}$ is a Nash equilibrium, from Proposition \ref{prop3} and equation (\ref{eq:strategies}), we know that they will play the joint action which is a Nash equilibrium for the rest of the game. 
 
Moreover in the case of the initial beliefs of the players are such that their initial joint action $s_{0}$ is not a Nash equilibrium based on Proposition \ref{prop1} and Proposition \ref{prop2} after a finite number of iterations because the game has a better reply path the only player that can improve his payoff by changing his actions will choose a new action which will result in  a new joint action $s$. If this action is not the a Nash equilibrium then again after finite number of iterations the player who can improve his payoff will change action and a new joint action $s'$ will be played. Thus after the search of the vertices of a finite graph, and thus after a finite number of iterations, players will choose a joint action which is a Nash equilibrium.
% the player that we know that after a finite number of iterations $t_{1}$ with high probability players at least one player will not change his action simultaneously with the others, thus for a Player $i$, $s^{i}_{t1-1}=s^{i}_{t1}$. If the new joint action $s_{t1}$ is not a Nash equilibrium, then at least one of the other players will deviate. Based on Proposition \ref{prop2} and the fact that players estimate opponents' strategies independently at least one of them who will not change his action simultaneously with the others after $t2$ iterations which will result to a new joint action $s_{t2}$ that will improve the current utility. Eventually after a finite number of time steps, $T$, the process will end up in a pure Nash equilibrium. The maximum number of iterations that is needed is the cardinality of the joint action set multiplied with the total number of iterations that is needed in order not to have simultaneous changes, $\binom{n}{2}(t1+t2+t3+\ldots+T)$
\end{proof}

\section{Simulations to define algorithm parameters $\Xi$ and $Z$.}
\label{parameters}
The covariance matrix of the state space error $\Xi=qI$ and the measurement error $Z=rI$ are two parameters that we should define in the beginning of the EKF fictitious play algorithm and they affect its performance.  Our aim is to find values, or range of values, of $q$ and $r$ that can efficiently track opponents' strategy when it smoothly or abruptly change, instead of choosing $q$ and $r$ heuristically for each opponent when we use the EKF algorithm. Nevertheless it is possible that for some games the results of the EKF algorithm will be improved for other combinations of $q$ and $r$ than the ones that we propose in this section. 

We examine the impact of EKF fictitious play algorithm parameters in its performance in the following two tracking scenarios.  In the first one a single opponent chooses his actions using a mixed strategy which changes smoothly and has a sinusoidal form over the iterations of the tracking scenario. In particular for $t=1, 2, \ldots, 100$ iterations of the game: $\sigma_{t}(1)=\frac{cos\frac{{2\pi t}}{n}+1}{2}=1-\sigma_t(2)$, where $n=100$. In the second toy example Player $i$'s opponent change his strategy abruptly and chooses action 1 with probability $\sigma_{t}^{2}(1)=1$ during the first 25 and the last 25 iterations of the game and for the rest iterations of the game $\sigma_{t}^{2}(1)=0$. The probability of the second action is calculated as: $\sigma_{t}^{2}(2)=1 - \sigma_{t}^{2}(1)$. 

We tested the performance of the proposed algorithm for the following range of parameters $10^{-4} \leq q \leq 1$ and  $10^{-4} \leq r \leq 1$. We repeated both 
examples 100 times for each of the combinations of $q$ and $r$. Each time we measured the absolute error of the estimated strategy against the real one. The combined average absolute error when both examples are considered is depicted on Figure \ref{fig:mseqr}. The darkest areas of the contour plot represent the areas where the average absolute error is minimised.

\begin{figure}
\centering
 \includegraphics[scale=0.5]{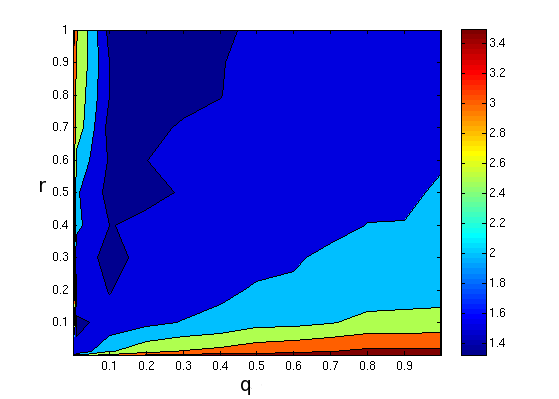}
\caption{Combined absolute error for both tracking scenarios. The range of both parameters, $q$ and $r$ is between $10^{-4}$ and $1$.}
\label{fig:mseqr}
\end{figure}

The average absolute error is minimised for a range of values of $q$ and $r$, that form two distinct areas. In the first area, the wide dark area of Figure \ref{fig:mseqr}, the range of $q$ and $r$ were $0.08 \leq q \leq 0.4$ and $0.2 \leq r \leq 1$ respectively. In the second area, the narrow dark area of Figure \ref{fig:mseqr}, the range of $q$ and $r$ were $0.001 \leq q \leq 0.025$ and $0.08 \leq r \leq 0.13$ respectively. The minimum error which we observed in  our simulations was in the narrow area and in particular when $\Xi=0.01I$ and $Z=0.1I$, where $I$ is the identical matrix. % These are the values of $W$ and $R$ that we will use in the simulations that follow in the next section.

\section{Simulation results}
\label{simulation}
This section is divided in two parts. The first part contains results of our simulations in two strategic form games and the second part contains the results we obtained in an ad-hoc sensor network surveillance problem. In all the simulations of this section we set the covariance matrix of the hidden and the observations state to $\Xi=0.01I$ and $Z=(0.1+\epsilon)I$ respectively, where $\epsilon \sim N(0,10^{-5})$ and $I$ is the identical matrix. 

\subsection{Simulations results in strategic form games}
In this section we compare the results of our algorithm with those of fictitious play in two
coordination games. These games are depicted in Tables \ref{tab:simcoord} and \ref{tab:3pclh}. The game that is depicted in Table \ref{tab:simcoord}, as it was described in Section \ref{theory} , is a simple coordination game with two pure Nash equilibria, its diagonal elements. Table \ref{tab:3pclh} presents an extreme version of the climbing hill game \citep{clh} in which three players must climb up a utility function in order to reach the Nash equilibrium where their reward is maximised.

\begin{table*}
%\flushleft
\centering
%\hspace*{-0.5cm}
\begin{tabular}{c|c|c|c}
\begin{tabular}{c} ~\\ \footnotesize U\\ \footnotesize M\\ \footnotesize D\\ ~ \end{tabular}
&
\begin{tabular}{ccc}
\small U &\small M &\small D \\    \hline
 \small   0&\small 0&\small 0\\
  \small  0&\small 50&\small 40\\
 \small   0&\small 0&\small 30\\      \hline
 \small   &U&
 \end{tabular}
&
 \begin{tabular}{ccc}
\small U &\small M &\small D \\ \hline
\small -300&\small 70&\small 80\\
\small -300&\small 60&\small 0\\
\small 0&\small 0&\small 0\\     \hline \small &M&
\end{tabular}
 &
 \begin{tabular}{ccc}
 \small U &\small M & \small D \\    \hline
\small \bf{100} & \small -300&\small 90\\
\small 0& \small 0& \small 0\\
\small 0&\small 0&\small 0\\           \hline
 \small &D&
 \end{tabular}

\end{tabular}
 \caption{ Climbing hill game with three players.  Player 1 selects rows, Player 2 selects columns, and Player 3 selects the matrix. The global reward depicted in the matrices,
is received by all players. The unique Nash equilibrium is in bold}
 \label{tab:3pclh}
\end{table*}

We present the results of 50 replications of a learning episode of 50 iterations for each game. As it is depicted in Figures \ref{fig:res1} and \ref{fig:res2} the proposed algorithm performs better than fictitious play in both cases. In the simple coordination game that is shown in Table \ref{tab:simcoord}, the EKF fictitious play algorithm converges to one of the pure equilibria after a few iterations. On the other hand fictitious play is trapped in a limit cycle in all the replications where the initial joint action was not one of the two pure Nash equilibria. For that reason the players' payoff for all the iterations of the game was either 1 utility unit or 0 utility units depending to the initial joint action. In the climbing hill game, Table \ref{tab:3pclh} the proposed algorithm converges to the Nash equilibrium after $35$ iterations when fictitious play algorithm do not converge even after 50 iterations.

\begin{figure}
\centering
 \includegraphics[scale=0.5]{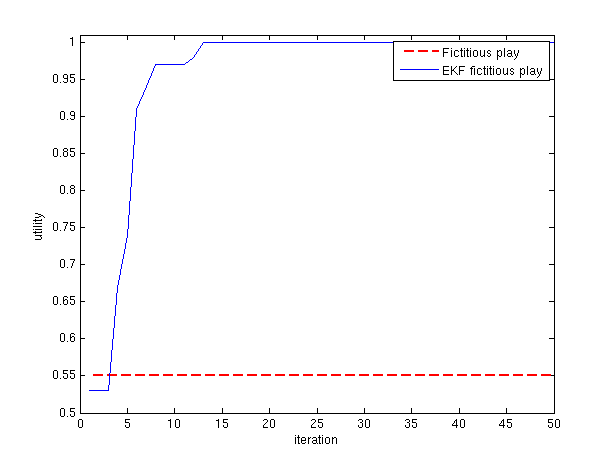}
\caption{Results of EKF and classic fictitious play in the simple coordination game of Table \ref{tab:simcoord}}
\label{fig:res1}
\end{figure}

\begin{figure}
\centering
 \includegraphics[scale=0.5]{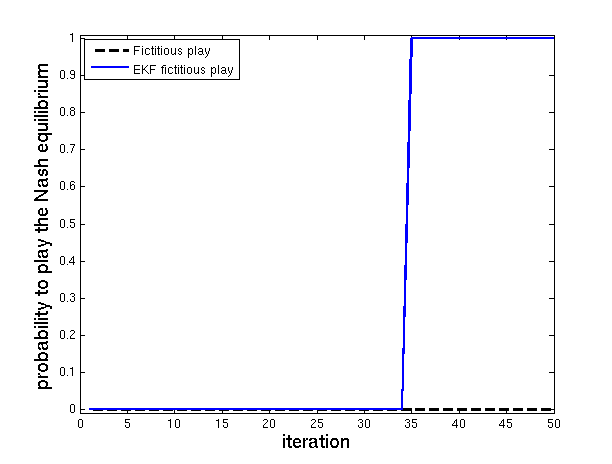}
\caption{Probability of playing the (U,U,D) equilibrium for the EKF fictitious play (solid line) and fictitious play (dash line) for the three player climbing hill game}
\label{fig:res2}
\end{figure}

\subsection{Ad-hoc sensor network surveillance problem.}
We compared the results of our algorithm against those of fictitious play in a coordination task of a power constrained sensor network, where sensors can be either in a sense or sleep mode \citep{faraneli,fict_ql}. When the sensors are in sense mode they can observe the events that occur in their range. During their sleep mode the sensors harvest the energy they need in order to be able function when they are in the sense mode. The sensors then should coordinate and choose their sense/sleep schedule in order to maximise the coverage of the events. This optimisation task can be cast as a potential game. In particular we consider the case where $\mathbb{I}$ sensors are deployed in an area where $E$ events occur. If an event $e$, $e \in E$, is observed from the sensors then it produce some utility $V_e$. Each of the sensors $i=1,\ldots,\mathbb{I}$ should choose an action $s^{i}=j$, from one of the $j=1,\ldots,J$ time intervals which they can be in sense mode. Each sensor $i$ when it is in sense mode can observe 
an event $e$, if it is in its sense range, with probability $p_{ie}=\frac{1}{d_{ie}}$, where $d_{ie}$ is the distance between the sensor $i$ and the event $e$. We assume that the probability each sensor has to observe an event is independent from the other sensors. If we denote as $i_{in}$ the sensors that are in sense mode when the event $e$ occurs and $e$ is in their sensing range, then we can write the probability an event $e$ to be observed from the sensors, $i_{in}$ as 
\begin{equation*}
1- \prod_{i \in i_{in}}{(1-p_{ie})}
\end{equation*}
The expected utility that is
produced from the event $e$ is the product of its utility $V_e$ and
the probability it has to be observed by the sensors, $i_{in}$ that are in sense mode when the event $e$ occurs and $e$ is in their sensing range. More formally we can express the utility that is
produced from an event $e$ as:
\begin{equation*}
U_{e}(s)=V_{e}(1- \prod_{i \in i_{in}}{(1-p_{ie}}))
\label{eq:target_utility}
\end{equation*}
The global utility is then the sum of the utilities that all events, $e \in E$, produce
\begin{equation*}
U_{global}(s)= \sum_{e}{U_{e}(s)}. \label{eq:global_utility}
\end{equation*}

Each sensor after each iteration of the game receives some utility which is based on the sensors and the events that are inside his communication and sense range respectively. For a sensor $i$ we denote $\tilde{e}$ the events that are in its sensing range and $\tilde{s}^{-i}$ the joint action of the sensors that are inside his communication range. The utility that sensor $i$ will receive if his sense mode is $j$ will be 
\begin{equation*}
U_{i}(s^{i}=j,\tilde{s}^{-i})= \sum_{\tilde{e}}{U_{\tilde{e}}(s^{i}=j,\tilde{s}^{-i})} \label{eq:individual_utility}.
\end{equation*}

We compared the performance of the two algorithms in 2 instances of the above scenario one with 20 and one with 50 sensors that are deployed in a unit square. In both instances sensors had to choose one time interval of the day that they will be in sense mode and use the rest time intervals to harvest energy. We consider cases where sensors had to choose their sense mode between 2, 3 and 4 available time intervals. Sensors are able to communicate with other sensors that are at most 0.6 distance units away, and can only observe events that are at most 0.3 distance units away. Moreover in both instances we assumed that 20 events took place in the unite square area. Those events were uniformly distributed in space and time, so an event could evenly appear in any point of the unit square area and it could occur at any time with the same probability. The duration of each event was uniformly chosen between (0-6] hours and each event had a value $V_{e} \in (0-1]$. 
Figures \ref{fig:res3} and \ref{fig:res4} depict the average results of 50 replications of the game for the two algorithms. For each instance, both algorithms run for 50 iterations. To be able to average across the 50 replications we normalise the utility of a replication by the global utility that the sensors will gain if they were only in sense mode during the whole day.

\begin{figure}[!ht]
\centering

\subfigure[Results when sensors have to choose between two time intervals.]{
   \includegraphics[scale =0.3] {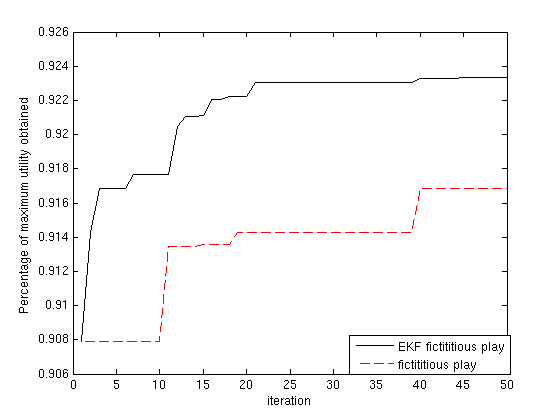}

 }
 \subfigure[Results when sensors have to choose between three time intervals.]{
   \includegraphics[scale=0.3] {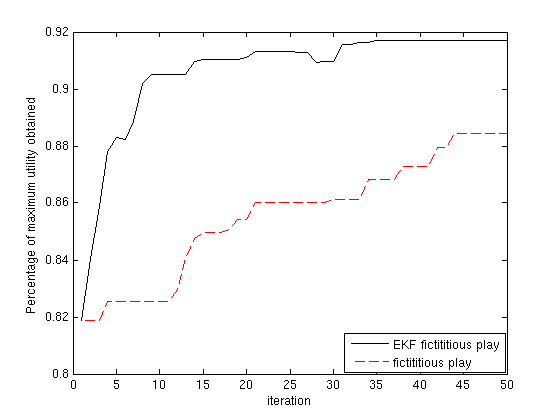}
   
 }

 \subfigure[Results when sensors have to choose between four time intervals.]{
   \includegraphics[scale =0.3] {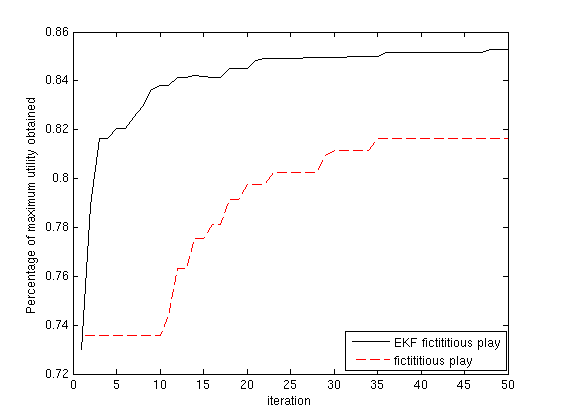}

 }
\caption{Results of the instance where 20 sensors should coordinate for both algorithms. The results of EKF fictitious play are the solid lines and the results of the classic fictitious play are the dash lines. The horizontal axis of the figures depict the iteration of the game and the vertical axis the global utility as a percentage of the global utility of the system in the case that sensors were always in sense mode.}
\label{fig:res3}
\end{figure}

\begin{figure}[!ht]
\centering

\subfigure[Results when sensors have to choose between two time intervals.]{
   \includegraphics[scale =0.3] {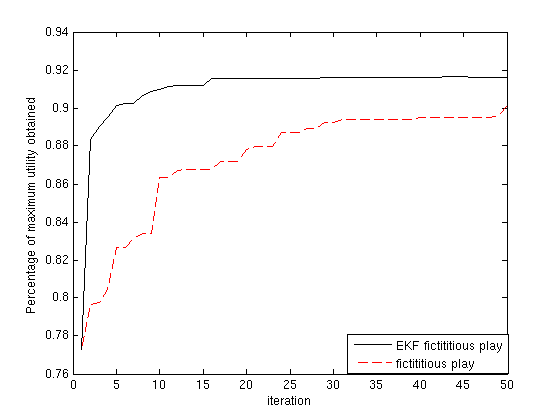}
  
 }
 \subfigure[Results when sensors have to choose between three time intervals.]{
   \includegraphics[scale=0.3] {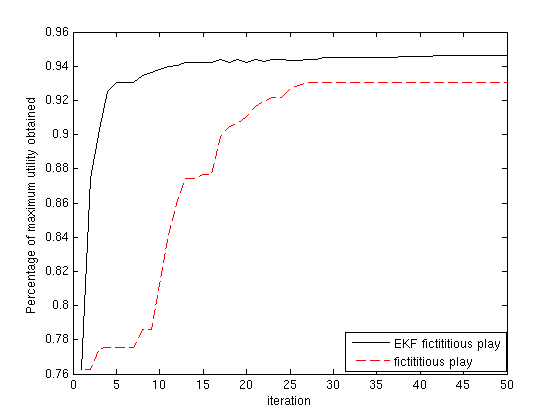}
   
 }

 \subfigure[Results when sensors have to choose between four time intervals.]{
   \includegraphics[scale =0.3] {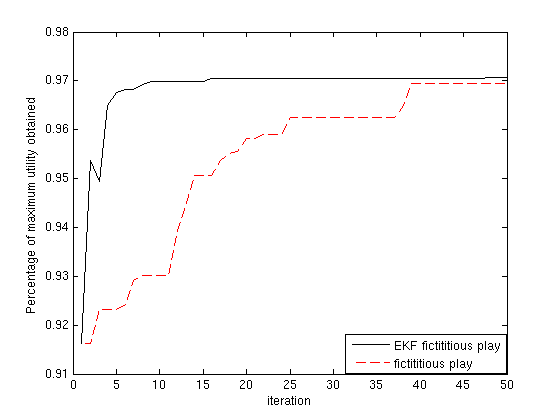}
   
 }
\caption{Results of the instance where 50 sensors should coordinate for both algorithms. The results of EKF fictitious play are the solid lines and the results of the classic fictitious play are the dash lines. The horizontal axis of the figures depict the iteration of the game and the vertical axis the global utility as a percentage of the global utility of the system in the case that sensors were always in sense mode.}
\label{fig:res4}
\end{figure}

As we observe in Figures \ref{fig:res3} and \ref{fig:res4} EKF fictitious play converges to a stable joint action faster than the fictitious play algorithm. In particular on average the EKF fictitious play algorithm needed 10 ``negotiation'' steps between the sensors in order to reach a stable joint action, when fictitious ply needed more than 25. Moreover the classic fictitious play algorithm was always resulted in joint actions with smaller reward than the proposed algorithm. 

\section{Conclusion}
We have introduced a variation of fictitious play that uses Extended Kalman filters to predict opponents' strategies. This variation of fictitious play addresses the implicit assumption of the classic algorithm that opponents use the same strategy in every iteration of the game. 

We showed that, for $2 \times 2$ games with at least one pure Nash equilibrium, EKF fictitious play converges in the pure Nash equilibrium of the game. More over the proposed algorithm converges in games with a better reply path, like potential games, and $n$ players that have 2 available actions. 

EKF fictitious play performed better than the classic algorithm algorithm in the strategic form games and the ad-hoc sensor network surveillance problem we simulated. 
Our empirical observations indicate that EKF fictitious play converges to a solution that is better than the classic algorithm and needs only a few iterations to reach that solution. Hence by slightly increasing the computational intensity of fictitious play less communication is required between agents to quickly coordinate on a desired solution.

\section{Acknowledgements}
This work is supported by The Engineering and Physical Sciences Research Council EPSRC (grant number EP/I005765/1).

\appendix
\section{Proof of Proposition 1}
\label{append1}
We will base the proof of Proposition \ref{prop1} on the properties of EKF when they used to estimate opponent's strategy with two available actions. If player $i$'s opponent has two available actions $1$ and $2$, then we can assume that at time $t-1$ Player $i$ maintains beliefs about his opponent's propensity, with mean $m_{t-1}$ and variance $P_{t-1}$. Moreover based on these estimations he chooses his strategy $\sigma_{t-1}$. At the prediction step of this process he uses the following equations to predict his opponent's propensity and choose an action using best response.

\begin{equation*}
m_{t}^{-} = \left(
\begin{array}{c}
m^{-}_{t-1}[1] \\
m^{-}_{t-1}[2] \\
\end{array} \right)
\end{equation*}

\begin{equation*}
P_{t}^{-} = \left(
\begin{array}{cc}
P^{-}_{t-1}[1,1] &P^{-}_{t-1}[1,2] \\
P^{-}_{t-1}[2,1]& P^{-}_{t-1}[2,2]  \\
\end{array} \right)+ qI
\end{equation*}
without loss of generality we can assume that his opponent in iteration $t$ chooses action 2. Then the update step will be :
\begin{equation*}
 v_t=z_{t}-h(m_{t}^{-})
\end{equation*}
since Players $i$'s opponent played action 2 and $h=\frac{exp(Q_{t}[s^{'}]/\tau)}{\sum_{\tilde{s} \in S} exp(Q_{t}[\tilde{s}]/\tau)}$ we can write $v_{t}$ and $H_{t}(m_{t}^{-})$ as:
\begin{align}
v_{t} = &\left(
\begin{array}{c}
0 \\
1 \\
\end{array} \right)-
\left(
\begin{array}{c}
\sigma_{t-1}(1) \\
1- \sigma_{t-1}(1) \\
\end{array} \right) \nonumber \\
 = &\left(
\begin{array}{c}
- \sigma_{t-1}(1) \\
\sigma_{t-1}(1) \\
\end{array}  \nonumber \right)
\end{align}

\begin{equation*}
H_{t}(m_{t}^{-}) = \left(
\begin{array}{cc}
a_{t} & -a_{t} \\
-a_{t}& a_{t}  \\
\end{array} \right)
\end{equation*}
where $a_{t}$ is defined $a_{t}=\sigma_{t-1}(1)\sigma_{t-1}(2)$. The estimation of \mbox{$S_{t}=H(m_{t}^{-})P_{t}^{-}H^{T}(m_{t}^{-})+Z$} will be: 
\begin{equation*}
S_{t} = a^{2}\left(
\begin{array}{cc}
b & -b \\
-b& b 
\end{array} \right)+Z
\end{equation*}
where $b=P^{-}_{t}[1,1]+P^{-}_{t}[2,2]-2P^{-}_{t}[1,2]$. The Kalaman gain, $K_{t}=P_{t}^{-}H^{T}(m_{t}^{-})S_{t}^{-1}$ can be written as 

\begin{equation*}
K_{t}= \frac{1}{2r b+r^2}\left(
\begin{array}{cc}
P_{t}^{-}[1,1] & k \\
k& P_{t}^{-}[2,2]  \\

\end{array} \right) \left(
\begin{array}{cc}
a_{t}& -a_{t} \\
-a_{t}&a_{t}  \\

\end{array} \right)\left(
\begin{array}{cc}
b+r & b\\
b& b+r \\
\end{array} \right)
\end{equation*}
up to a multiplicative constant we can write
\begin{equation*}
K_{1} \sim \left(
\begin{array}{cc}
c & -c\\
-d& d \\

\end{array} \right)
\end{equation*}
where $c=P_{t}^{-}[1,1]-P_{t}^{-}[1,2]$ and $d=P_{t}^{-}[2,2]-P_{t}^{-}[1,2]$.  The updates then for the mean and variance are:

\begin{align}
	m_{t}=&m_{t}^{-}+K_{t}v_{t}  \nonumber \\
	P_{t}=&P_{t}^{-}-K_{t}S_{t}K_{t}^{T} \nonumber
%\label{31}
\end{align}
The mean of the Gaussian distribution that is used to estimate opponent's propensities is:  
\begin{align}
\label{eq:mupd}
m_{t} = &\left(
\begin{array}{c}
m_{t}[1] \\
m_{t}[2] \\
\end{array} \right) 
 =\left(
\begin{array}{c}
m_{t}^{-}[1]-2\sigma(1)\frac{a(b-k)}{4a^{2}(b-k)+(r+\epsilon)}\\
m_{t}^{-}[2]+2\sigma(1)\frac{a(b-k)}{4a^{2}(b-k)+(r+\epsilon)}\\
\end{array} \right)
\end{align}

Based on the above we observe that $m_{t}(1)<m_{t-1}(1)$ and $m_{t}(2)>m_{t-1}(2)$ which completes the proof.

\section{Proof of Proposition 2}
\label{append2}
We consider $2 \times 2$ games with at least one pure Nash equilibrium. In the case that only one Nash equilibrium exists, a dominant strategy exists and thus one of the players will not deviate from this action. Hence we are interested in in $2\times2$ games with two pure Nash equilibria. Without loss of generality we consider a game with similar structure to the simple coordination game that is depicted in Table \ref{tab:simcoord}. with two equilibria, the joint actions in the diagonal of the payoff matrix, $(U,L)$ and $(D,R)$. We will present calculations for Player 1,but the same results hold also for Player 2. We define $\lambda$ as the necessary confidence level that Player 1's estimation of $\sigma_{t}(L)$ should reach in order to choose action $U$. Hence we Player 1 will choose $D$ if: 
\begin{align}
\sigma_{t}(1) &>\lambda  \Leftrightarrow \nonumber \\
\frac{exp(m_{t}^{-}[1])}{exp(m_{t}^{-}[1])+exp(m_{t}^{-}[2])} & > \lambda \Leftrightarrow \nonumber \\
 m_{t}^{-}[1] &> \ln(\frac{\lambda}{1-\lambda}) +m_{t}^{-}[2]  \Leftrightarrow \nonumber \\
 m_{t-1}[1] &> \ln(\frac{\lambda}{1-\lambda}) +m_{t-1}[2]  \nonumber
\end{align}

In order to prove Proposition \ref{prop2}, we need to show that when a player changes his action his opponent will change his action at the same iteration with probability less than 1. In the case where at time $t-1$ the joint action of the players is $U,R$ then Player $1$ believes that his opponent will play $L$, while he observing him playing $R$. Assume that Player 2's beliefs about Player 1's strategies has reached the necessary confident level about Players 1's strategy and at iteration $t$ he will change his action from $R$ to $L$. Player 1 will also change his action at the same time if  

\begin{equation*}
  m_{t-1}[2] > \ln(\frac{1-\lambda}{\lambda}) +m_{t-1}[1]
\end{equation*}
We want to show that players will not change actions simultaneously with probability 1. Hence it is enough to show that 

\begin{equation}
 Prob( m_{t-1}[1] > \ln(\frac{\lambda}{1-\lambda}) +m_{t-1}[2]  )>0
\label{pithan1212}
\end{equation}
We can  replace $m_{t-1}[1]$ and $m_{t-1}[2]$ with their equivalent from (\ref{eq:mupd}) and write:

\begin{eqnarray}
m_{t}^{-}[1]-2\sigma(1)\frac{a(b-k)}{4a^{2}(b-k)+(r+\epsilon)} & > &\ln( \frac{\lambda}{1-\lambda})+ m_{t}^{-}[2]+2\sigma(1)\frac{a(b-k)}{4a^{2}(b-k)+(r+\epsilon)} \Leftrightarrow \nonumber  
\end{eqnarray}

\begin{eqnarray}
-4\sigma(1)\frac{a(b-k)}{4a^{2}(b-k)+(r+\epsilon)} &> & \ln(\frac{\lambda}{1-\lambda})+ m_{t}^{-}[2] - m_{t}^{-}[1] \Leftrightarrow \nonumber \\
\frac{a(b-k)}{4a^{2}(b-k)+(r+\epsilon)} &<&\frac{\ln(\frac{\lambda}{1-\lambda})+ m_{t}^{-}[2] - m_{t}^{-}[1]}{-4\sigma(1)} \nonumber
\end{eqnarray}
Solving this with respect to $\epsilon$ we have 

\begin{equation*}
\epsilon >\frac{a(b-k)\sigma(1)}{\ln(\frac{\lambda}{1-\lambda})+ m_{t}^{-}[2] - m_{t}^{-}[1]}-a^{2}(b-k) -r 
\end{equation*}
Thus we can write (\ref{pithan1212}) as:  
\begin{equation}
 Prob(\epsilon >\frac{a(b-k)\sigma(1)}{\ln(\frac{\lambda}{1-\lambda})+ m_{t}^{-}[2] - m_{t}^{-}[1]}-a^{2}(b-k) -r  )>0
\label{pithan1111}
\end{equation}
Since $\epsilon$ is a Gaussian white noise (\ref{pithan1111}) is always true.  

We also consider the case where at time $t-1$ the joint action of the players is $D,L$ then Player 1 believes that his opponent will play $R$, while he observing him playing $L$. Assume that Player 2's beliefs about Player 1's strategies has reached the necessary confident level and at $t$ he will change his action from $L$ to $R$. Player 1 will also change his action at the same time if  

\begin{equation*}
   m_{t-1}[1] > \ln(\frac{\lambda}{1-\lambda}) +m_{t-1}[2] 
\end{equation*}
We want to show that Players will not change actions simultaneously with probability 1. Hence it is enough to show that 

\begin{equation}
 Prob(m_{t-1}[2] > \ln(\frac{1-\lambda}{\lambda}) +m_{t-1}[1])>0
\label{pithan}
\end{equation}
We can  rewrite (\ref{pithan}) using the results we obtained for  $m_{t-1}[1]$ and $m_{t-1}[2]$ in (\ref{eq:mupd}) again as

\begin{equation}
 Prob(\epsilon >\frac{a(b-k)\sigma(1)}{\ln(\frac{\lambda}{1-\lambda})+ m_{t}^{-}[2] - m_{t}^{-}[1]}-a^{2}(b-k) -r  )>0
\label{pithan1}
\end{equation}
Since $\epsilon$ is a Gaussian white noise (\ref{pithan1}) is always true.  

If we define $\xi_{t}$ the event that both players change their action at time $t$ simultaneously, and assume that the two players have change their actions simultaneously at the following iterations $t_{1}, t_{2}, \ldots, t_{t}$, then the probability that they will also change their action simultaneously at time $t_{T+1}$, $P(\xi_{t_1},\xi_{t_2}, \ldots, \xi_{t_T}, \xi_{t_T+1})$ is almost zero for large but finite $T$.
\bibliographystyle{model4-names}
\bibliography{ekf}
\end{document}